\documentclass[reqno]{amsart}

\usepackage{amsthm}
\usepackage{amsmath}
\usepackage{latexsym}
\usepackage{amsfonts}
\usepackage{amssymb}
\usepackage{color}
\usepackage{bbm,dsfont}
\usepackage{graphicx}
\usepackage{subfigure}
\usepackage{mathbbol}
\usepackage{hyperref}
\usepackage{enumerate}
\usepackage[foot]{amsaddr}

\newtheorem{proposition}{Proposition}
\newtheorem{proposition?}{Proposition?}

\newtheorem{corollary}{Corollary}

\theoremstyle{definition}

\newtheorem{example}{Example}
\newtheorem{definition}{Definition}

\newcommand{\complex}{\mathbb C} 
\newcommand{\integer}{\mathbb Z} 
\newcommand{\half}{\tfrac{1}{2}} 

\newcommand{\hi}{\mathcal{H}} 
\newcommand{\hik}{\mathcal{K}} 
\newcommand{\lin}{\mathcal{L}}
\newcommand{\lh}{\mathcal{L(H)}} 
\newcommand{\lk}{\mathcal{L(K)}} 
\newcommand{\sh}{\mathcal{S(H)}} 
\newcommand{\ip}[2]{\left\langle\,#1\,|\,#2\,\right\rangle} 
\newcommand{\kb}[2]{|#1\rangle\langle#2|} 
\newcommand{\no}[1]{\left\|#1\right\|} 
\newcommand{\tr}[1]{\mathrm{tr}\left[#1\right]} 
\newcommand{\ptr}[2]{\mathrm{tr}_{#1}[#2]} 
\newcommand{\id}{\mathbbm{1}} 

\newcommand{\ltwo}[1]{\ell^2(#1)} 

\newcommand{\obsin}{\mathfrak{O}(\hin)}
\newcommand{\F}{\mathsf{F}}
\newcommand{\G}{\mathsf{G}}
\newcommand{\M}{\mathsf{M}}
\newcommand{\N}{\mathsf{N}}
\newcommand{\X}{\mathsf{X}}
\newcommand{\Y}{\mathsf{Y}}

\newcommand{\I}{\mathcal{I}}

\newcommand{\chan}{\mathfrak{C}}
\newcommand{\chansub}{\mathfrak{D}}

\newcommand{\chin}{\chan(\hin)}

\newcommand{\Cc}{\mathcal{C}}

\newcommand{\pleq}{\preceq}

\newcommand{\hin}{\hi^{in}} 
\newcommand{\hout}{\hi^{out}} 
\newcommand{\state}[1]{\mathcal{S}(#1)}
\newcommand{\lhin}{\mathcal{L}(\hin)}
\newcommand{\lhout}{\mathcal{L}(\hout)}

\begin{document}

\title[]{Incompatibility of quantum channels}

\author{Teiko Heinosaari$^\clubsuit$}
\address{$\clubsuit$ Turku Centre for Quantum Physics, Department of Physics and Astronomy, University of Turku, FI-20014, Finland}

\author{Takayuki Miyadera$^\spadesuit$}
\address{$\spadesuit$ Department of Nuclear Engineering, Kyoto University, 6158540 Kyoto, Japan}

\begin{abstract}
Two quantum channels are called compatible if they can be obtained as marginals from a single broadcasting channel; otherwise they are incompatible. 
We derive a characterization of the compatibility relation in terms of concatenation and conjugation, and we show that all pairs of sufficiently noisy quantum channels are compatible.
The complement relation of incompatibility can be seen as a unifying aspect for several important quantum features, such as impossibility of universal broadcasting and unavoidable measurement disturbance.
We show that the concepts of entanglement breaking channel and antidegradable channel can be completely characterized in terms compatibility.
\end{abstract}

\maketitle


\section{Introduction}

One of the fundamental features of quantum theory is that not all observables are jointly measurable. 
This observation goes back to the Heisenberg's uncertainty principle and Bohr's notion of complementarity, and has since then been studied extensively.
Two observables that do not have a joint measurement are called incompatible. 
Recently incompatibility of observables has been formulated and studied in general operational theories \cite{BuHeScSt13,StBu14,Banik15,Plavala16,SeReChZi16}, hence opening the possibility to compare features of incompatibility in quantum theory to other operational theories.
Interestingly, quantum theory contains maximally incompatible pairs of observables, but only in an infinite dimensional Hilbert space \cite{HeScToZi14}.
There appear to be diverse aspects of quantum incompatible that urge for further investigation.

Incompatibility can be defined not only for observables but also for channels \cite{HeMiZi16}.
We recall this definition and show that it reduces to the joint measurability if channels are of the measurement form, hence the definition is, indeed, a natural generalization of joint measurability of quantum observables.
As pointed out in \cite{HeMiZi16}, the compatibility relation is linked to the concatenation of channels. 
We develop this idea further and characterize the compatibility relation on channels in terms of concatenation and conjugation. 
The central aim of this paper is to demonstrate the broad applicability of the concepts of compatibility and incompatibility, and reveal their connections to various features of quantum information processing.
In particular, the impossibility of universal broadcasting \cite{Barnumetal96} and the unavoidability of measurement disturbance \cite{HeMi13} appear naturally in this framework.

We prove that the concepts of completely depolarizing channel, entanglement breaking channel \cite{HoShRu03} and antidegradable channel \cite{CuRuSm08} can be completely characterized in terms of compatibility.
The introduced framework allows us also to generalize the notion of incompatibility breaking channels \cite{HeKiReSc15} and we show that a channel destroying the incompatibility between any finite set of channels is entanglement breaking.

The outline of this paper is as follows. 
In Sec. \ref{sec:channels} we recall the relevant basic concepts, including the concatenation preorder of quantum channels.
In Sec. \ref{sec:incompatibility} we formulate the incompatibility of channels and study its consequences. 
Finally, in Sec. \ref{sec:special} we show how certain special classes of channels can be characterized in terms of their compatibility properties.

We will restrict to finite dimensional Hilbert spaces.
We denote by $\lh$ the vector space of linear operators on a Hilbert space $\hi$, and by $\sh$  the states on $\hi$, i.e., positive operators of trace one.

\section{Channels, observables and concatenation}\label{sec:channels}

\subsection{Quantum channels}

Let $\hi$ and $\hik$ be Hilbert spaces.
A \emph{quantum channel} is a completely positive linear map $\Lambda: \lh \to \lk$ that is unital, i.e., $\Lambda(\id_\hi)=\id_\hik$. 
This mathematical description is the Heisenberg picture of a quantum channel and will be the most suitable for our investigation.
The physical meaning of a channel $\Lambda$ is, perhaps, more evident when we look at its dual action on states. 
The Schr\"odinger picture of $\Lambda$ is the map $\Lambda^*: \lk \to \lh$ determined by the formula
\begin{equation}
\tr{\varrho \Lambda(T)} = \tr{ \Lambda^*(\varrho) T} \, , 
\end{equation}
required to hold for all states $\varrho \in \state{\hik}$ and operators $T \in \lin(\hi)$.
In the Schr\"odinger picture a quantum channel is a completely positive and trace preserving linear map on Hilbert space operators. 
A channel $\Lambda^*$ is fully specified by its action on the set of states, so we often write it as a map on states to further emphasize the use of the Schr\"odinger picture.
We will use the symbol $*$ in the superscript to denote the Schr\"odinger picture of a channel. 

We will denote by $\hin$ and $\hout$ the input and output Hilbert spaces in the Schr\"odinger picture, respectively.
The output space $\hout$ can be different from the input space $\hin$.
For instance, a channel that adds another system in a fixed state $\eta$, i.e., the map $\varrho \mapsto \varrho \otimes \eta$ is a valid channel.
We will mostly focus on channels that have the same fixed input space $\hin$ but arbitrary (finite dimensional) output space.
We denote by $\chan(\hin)$ this set of channels.

\subsection{Quantum observables}

A \emph{quantum observable} is commonly described as a positive operator valued measure (POVM).
We will assume that there are finite number of possible measurement outcomes, so it is possible and convenient to define an observable as a function $\M:x\mapsto \M(x)$ from a finite set of measurement outcomes $\Omega_\M \subset \integer$ to the set of positive operators on an input Hilbert space $\hin$.
This function must satisfy the normalization constraint $\sum_{x\in\Omega_\M} \M(x)=\id$, where $\id$ is the identity operator on $\hin$.
The probability of obtaining a measurement outcome $x$ for an input state $\varrho$ is $\tr{\varrho\M(x)}$.
We denote by $\obsin$ the set of all observables on $\hin$.

We will write and think of an observable as a special kind of channel. 
First, for a finite set $\Omega\subset\integer$, we denote by $\ltwo{\Omega}$ the Hilbert space of functions $f:\Omega \to \complex$.
The inner product of two functions $f$ and $g$ is
\begin{equation*}
\ip{f}{g} = \sum_x \overline{f(x)} g(x) \, .
\end{equation*}
For each $x\in \Omega$, we denote by $\delta_x$ the Kronecker function of $x$, i.e., $\delta_x(x)=1$ and $\delta_x(y)=0$ for $y\neq x$.
The set $\{ \delta_x : x \in \Omega \}$ is an orthonormal basis of $\ltwo{\Omega}$.
In particular, the dimension of $\ltwo{\Omega}$ is the order of $\Omega$.
For each observable $\M$, we define a channel $\Gamma_\M: \lin(\ltwo{\Omega_\M}) \to \mathcal{L}(\hin)$ as
\begin{equation}
\Gamma_\M(A) = \sum_x \ip{\delta_x}{A \delta_x} \M(x) \, .
\end{equation}
In the Schr\"odinger picture this channel reads
\begin{equation}
\Gamma^*_\M(\varrho) = \sum_x \tr{\varrho \M(x)} \kb{\delta_x}{\delta_x} \, ,
\end{equation}
hence, this is a channel that writes the measurement probabilities $\tr{\varrho \M(x)}$ into orthogonal pure states. 
The essential point is that orthogonal pure states are perfectly distinguishable, so the measurement outcome distribution can be recovered from the output state $\Gamma^*_\M(\varrho)$.

\subsection{Concatenation preorder}

Suppose we have two channels $\Lambda_1:\lin(\hout_1) \to \lin(\hin_1)$ and $\Lambda_2:\lin(\hout_2) \to \lin(\hin_2)$ such that $\lin(\hin_1) = \lin(\hout_2)$.
Then the functional composition $\Lambda_2 \circ \Lambda_1$ is defined and it is a channel from $\lin(\hout_1)$ to $\lin(\hin_2)$. 
Physically the composition corresponds to a sequential implementation of these two channels and we call the new channel a \emph{concatenation} of $\Lambda_1$ and $\Lambda_2$. 

\begin{definition}
For two channels $\Lambda_1$ and $\Lambda_2$, we denote $\Lambda_1 \pleq \Lambda_2$ if $\Lambda_1 =  \Lambda_2 \circ \Theta$ for some channel $\Theta$.
We also denote $\Lambda_1 \simeq \Lambda_2$ if both $\Lambda_1 \pleq \Lambda_2$ and $\Lambda_2 \pleq \Lambda_1$ hold.
\end{definition}

The binary relation $\pleq$ is reflexive (i.e. $\Lambda \pleq \Lambda$) and transitive (i.e. $\Lambda_1 \pleq \Lambda_2 \pleq \Lambda_3$ implies $\Lambda_1 \pleq \Lambda_3$), hence it is a preorder on $\chan(\hin)$.
It fails to be a partial order since it is not antisymmetric; there are pairs of channels $\Lambda_1$ and $\Lambda_2$ such that $\Lambda_1 \simeq \Lambda_1$ but $\Lambda_1 \neq \Lambda_2$. 
We say that two channels $\Lambda_1$ and $\Lambda_2$ satisfying $\Lambda_1 \simeq \Lambda_1$ are \emph{(concatenation) equivalent}.
In the Schr\"odinger picture the order of concatenation is the opposite to the that of Heisenberg picture, i.e., 
\begin{equation*}
 (\Lambda_2 \circ \Lambda_1 )^* = \Lambda^*_1 \circ \Lambda^*_2 \, .
\end{equation*}

In the following we show that for two channels related to observables the concatenation preorder is equivalent to the post-processing relation \cite{MaMu90a}.

\begin{proposition}
Let $\M$ and $\N$ be two observables. 
The following are equivalent:
\begin{enumerate}[(i)]
\item $\Gamma_\N \pleq \Gamma_\M$
\item $\N$ is a post-processing of $\M$, i.e., there is a stochastic matrix $\nu$ such that
\begin{equation}\label{eq:post}
\N(x) = \sum_y \nu_{xy} \M(y) \, .
\end{equation}
\end{enumerate}
\end{proposition}

\begin{proof}
(i)$\Rightarrow$(ii): By the assumption there exists a channel $\Theta$ such that $\Gamma_\N = \Gamma_\M \circ \Theta$.
For each $x\in\Omega_\M$, $y\in\Omega_\N$, we define $\nu_{xy}$ as
\begin{equation}
\nu_{xy} = \ip{\delta_y}{\Theta(\kb{\delta_x}{\delta_x})\delta_y} \, .
\end{equation}
It follows from the positivity and unitality of $\Theta$ that $\nu$ is a stochastic matrix.
The equality $\Gamma_\N(\kb{\delta_x}{\delta_x}) = (\Gamma_\M \circ \Theta)(\kb{\delta_x}{\delta_x})$ then implies \eqref{eq:post}.

(ii)$\Rightarrow$(i): For each $x\in\Omega_\M$, $y\in\Omega_\N$, we define an operator $K_{xy}$ as 
\begin{equation*}
K_{xy} = \sqrt{\nu_{xy}} \kb{\delta_x}{\delta_y} \, .
\end{equation*}
Then we define a map $\Theta$ as
\begin{equation}
\Theta(A) = \sum_{x,y} K_{xy}^\ast A K_{xy} \, .
\end{equation}
This is a Kraus operator-sum form, so $\Theta$ is completely positive.
It is direct to verify that $\Gamma_\N = \Gamma_\M \circ \Theta$.
\end{proof}

If $\M$ is an observable and $\Lambda$ is a channel, then we denote by $\Lambda(\M)$ the observable defined as
\begin{equation}
\Lambda(\M)(x) := \Lambda(\M(x)) \, .
\end{equation}
Hence, $\Lambda \circ \Gamma_\M = \Gamma_{\Lambda(\M)}$.

\subsection{Tensor product of channels}

While the concatenation corresponds to a sequential implementation of two channels, we can also implement two channels in parallel.
The essential difference is that in the parallel implementation one needs two input systems instead of one. 

Suppose we have two linear maps $\Lambda_1:\lin(\hout_1) \to \lin(\hin_1)$ and $\Lambda_2:\lin(\hout_2) \to \lin(\hin_2)$.
For all $A \in \lin(\hout_1)$ and $B�\in \lin(\hout_2)$, we denote
\begin{equation}\label{eq:tensor}
\Lambda_1 \otimes \Lambda_2 (A \otimes B) := \Lambda_1(A) \otimes \Lambda_2 (B) \, .
\end{equation}
Since the product operators $A \otimes B$ span the vector space $\lin (\hout_1 \otimes \hout_2)$, the formula \eqref{eq:tensor} determines a linear map $\Lambda_1 \otimes \Lambda_2$ from $\lin (\hout_1 \otimes \hout_2)$ to $\lin (\hin_1 \otimes \hin_2)$.
The map $\Lambda_1 \otimes \Lambda_2$ is called the tensor product of $\Lambda_1$ and $\Lambda_2$.
Clearly, the tensor product of two channels is a channel.

\subsection{Conjugate channel}

We recall that by the \emph{Stinespring dilation theorem} any channel $\Lambda:\lin(\hout)\to\lin(\hin)$ can be written in the form
\begin{equation}\label{eq:stinespring-h}
\Lambda(A) =V^*(A\otimes \id_\hik)V \, ,
\end{equation}
where $\hik$ is a Hilbert space and $V:\hin\to\hout\otimes \hik$ is an isometry, i.e., $V^*V=\id$ (see e.g. \cite{CBMOA03}).
The pair $(V, \hik)$ is called a \emph{Stinespring representation} for $\Lambda$.
A Stinespring representation $(V, \hik)$ for $\Lambda$ is called \emph{minimal} 
if the set $(\lin(\hout) \otimes \id) V \hin$ is dense in $\hout \otimes \hik$. 
Every channel has a minimal Stinespring representation, and if $\hin$ and $\hout$ are finite dimensional, then also $\hik$ is finite dimensional.
All Stinespring representations of $\Lambda$ can be obtained from a minimal one $(V, \hik)$ as follows:
for a Stinespring representation $(V',\hik')$ of $\Lambda$, there is an isometry $W: \hik \to \hik'$ 
such that 
\begin{equation}\label{eq:minimal}
V'= (\id_{\hin} \otimes W) V \, .
\end{equation} 
In addition, $(V', \hik')$ is minimal if and only if $W$ is unitary. 

The formula \eqref{eq:stinespring-h} gives rise to another channel $\bar{\Lambda}:\lin(\hik)\to\lin(\hin)$, defined as
\begin{equation}\label{eq:conjugate-h}
\bar{\Lambda}(B) =V^*(\id_{\hin} \otimes B)V \, ,
\end{equation}
and called a \emph{conjugate channel} of $\Lambda$.
The conjugate channel obviously depends on the used Stinespring representation of $\Lambda$, so the notation $\bar{\Lambda}$ should be used cautiously.
In the Schr\"odinger picture the formulas \eqref{eq:stinespring-h} and \eqref{eq:conjugate-h} read
\begin{equation}\label{eq:stinespring-s}
\Lambda^*(\varrho)=\ptr{\hik}{V\varrho V^*} \, , 
\end{equation}
and
\begin{equation}
\bar{\Lambda}^*(\varrho)=\ptr{\hout}{V\varrho V^*} \, ,
\end{equation}
so the conjugate channel is obtained when we trace over $\hout$ rather than $\hik$.
It is clear from the definition that a channel $\Lambda$ is a conjugate channel of its conjugate channel.
We will think the conjugacy as a symmetric relation in $\chin$.

The important fact for our following results is that \emph{all conjugate channels of a given channel $\Lambda$ are concatenation equivalent}.
To see this, let $\bar{\Lambda}$ be a conjugate channel constructed by using a minimal Stinespring representation 
$(V, \hik)$, and let $\bar{\Lambda}'$ be another conjugate channel related to a Stinespring representation $(V', \hik')$. 
Since $(V,\hik)$ is minimal, there exists an isometry $W:\hik \to \hik'$ satisfying \eqref{eq:minimal}.
We define a channel $\Theta:\lin(\hik')\to\lin(\hik)$ as $\Theta(A)= W^* A W$, and then $\bar{\Lambda}\circ \Theta= \bar{\Lambda}'$, showing that $\bar{\Lambda}' \pleq \bar{\Lambda}$.
On the other hand, fix a state $\eta\in\state{\hik}$ and define a channel $\Theta':\lin(\hik)\to\lin(\hik')$ as
\begin{equation}
\Theta'(A)= WA W^* + (\id - WW^*) \tr{\eta A} \, . 
\end{equation}
Then  $\bar{\Lambda} = \bar{\Lambda}'\circ \Theta'$, hence $\bar{\Lambda} \pleq \bar{\Lambda}'$, and therefore $\bar{\Lambda} \simeq \bar{\Lambda}'$.

The following result will be used several times later \cite[Theorem 2]{BeOr11}.

\begin{proposition}\label{prop:c-inverts-order}
Let $\Lambda_1,\Lambda_2\in\chin$.
$\Lambda_1 \pleq \Lambda_2$ if and only if $\bar{\Lambda}_2 \pleq \bar{\Lambda}_1$. 
\end{proposition}

\begin{proof}
Let us assume that $\Lambda_1 \pleq \Lambda_2$, so there exists a channel $\Theta$ such that $\Lambda_1 =\Lambda_2 \circ \Theta$. 
We fix minimal Stinespring representations for 
$\Lambda_1, \Lambda_2$ and $\Theta$, so that 
$\Lambda_1(A)=V_1^*(A\otimes \id)V_1$, 
$\Lambda_2(B) =V_2^*(B\otimes \id)V_2$ and 
$\Theta(A)=V_{\Theta}^*(A\otimes \id)V_{\Theta}$.   
From $\Lambda_1 =\Lambda_2 \circ \Theta$ follows that
\begin{align}
V_1^*(A\otimes \id)V_1 =
V_2^*(V_\Theta^*\otimes \id)(A\otimes \id \otimes \id)(V_\Theta\otimes \id)V_2
\end{align}
for all $A\in \mathcal{L}(\hi)$. 
The minimality of $V_1$ implies that there exists an isometry $W$ satisfying 
\begin{align}
(V_\Theta\otimes \id)V_2 =(\id \otimes W)V_1.
\end{align}
The conjugate channel of $\Lambda_2$ satisfies 
for all $B\in\lin(\hik_2)$, 
\begin{eqnarray*}
\bar{\Lambda}_2(B) &=& 
V_2^*(\id \otimes B) V_2 =
V_2^* (V_\Theta^* \otimes \id)
(\id \otimes \id \otimes B)(V_\Theta\otimes \id) V_2 
\\
&=& 
V_1^*(\id \otimes W^*)(\id \otimes \id \otimes B)(\id \otimes W)V_1 
= 
V_1^* (\id \otimes W^*(\id \otimes B)W)V_1 \\
&=&  \bar{\Lambda}_1 \circ \Theta_W(B) \, , 
\end{eqnarray*}
where $\Theta_W :\lin(\hik_2)\to \lin(\hik_1)$
is a channel defined by 
\begin{equation}
\Theta_W(B) =W^*(\id \otimes B)W \, .
\end{equation}
Thus, we conclude that $\bar{\Lambda}_2 \pleq \bar{\Lambda}_1$. 

If we start from the assumption $\bar{\Lambda}_2 \pleq \bar{\Lambda}_1$, then the previous calculations show that $\bar{\bar{\Lambda}}_1 \pleq \bar{\bar{\Lambda}}_2$.
Since $\Lambda_1 \simeq \bar{\bar{\Lambda}}_1$ and $\Lambda_2 \simeq \bar{\bar{\Lambda}}_2$, it follows that $\Lambda_1 \pleq \Lambda_2$. 
\end{proof}

\section{Incompatibility and its consequences}\label{sec:incompatibility}

\subsection{Definition and basic properties}

Let us consider a channel $\Lambda^*$ that has an input space $\hin$ and the output space is a tensor product $\hout = \hout_1 \otimes \hout_2$.
This kind of channel is called a \emph{quantum broadcast channel} \cite{YaHaDe11}.
By concatenating $\Lambda^*$ with the partial traces on subsystems we get two channels $\Lambda^*_1$ and $\Lambda^*_2$,
\begin{equation}\label{eq:marginals-s}
\Lambda^*_1(\varrho) = \ptr{\hout_2}{\Lambda(\varrho)} \, , \quad \Lambda^*_2(\varrho) = \ptr{\hout_1}{\Lambda(\varrho)} \, .
\end{equation}
This corresponds to ignoring one part of the output.
In the Heisenberg picture the marginal conditions in \eqref{eq:marginals-s} read
\begin{equation}\label{eq:marginals-h}
\Lambda_1(A) = \Lambda(A \otimes \id) \, , \quad \Lambda_2(B) = \Lambda(\id \otimes B) \, ,
\end{equation}
required to hold for all $A\in\lin(\hout_1)$ and $B\in\lin(\hout_2)$.

\begin{definition}\label{def:compatible}
Let $\Lambda_1:\lin(\hout_1) \to \lin(\hin)$ and $\Lambda_2:\lin(\hout_2) \to \lin(\hin)$ be two channels.
If there exists a channel $\Lambda:\lin(\hout_1 \otimes \hout_2)\to\lin(\hin)$ such that \eqref{eq:marginals-h} holds for all $A\in\lin(\hout_1)$ and $B\in\lin(\hout_2)$, then $\Lambda_1$ and $\Lambda_2$ are \emph{compatible} and $\Lambda$ is their \emph{joint channel}.
Otherwise $\Lambda_1$ and $\Lambda_2$ are \emph{incompatible}.
\end{definition}

\begin{figure}
    \subfigure[]
    {
         \includegraphics[width=5cm]{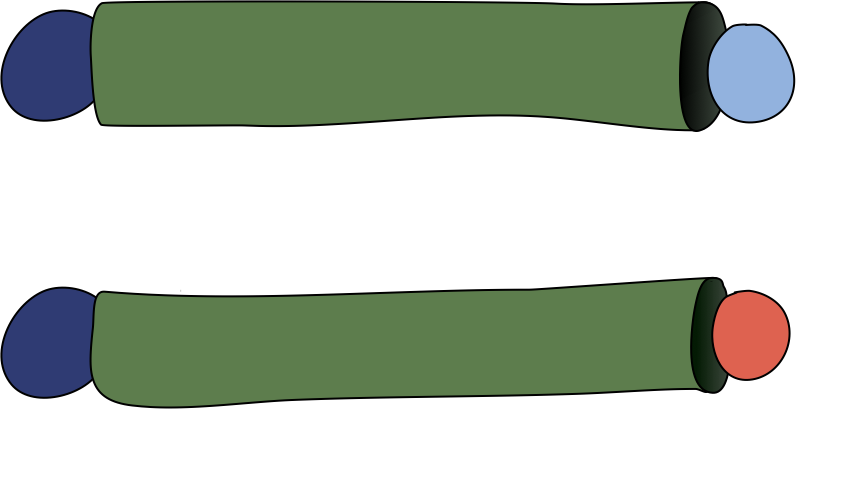}      
         }
    \subfigure[]
    {
        \includegraphics[width=5cm]{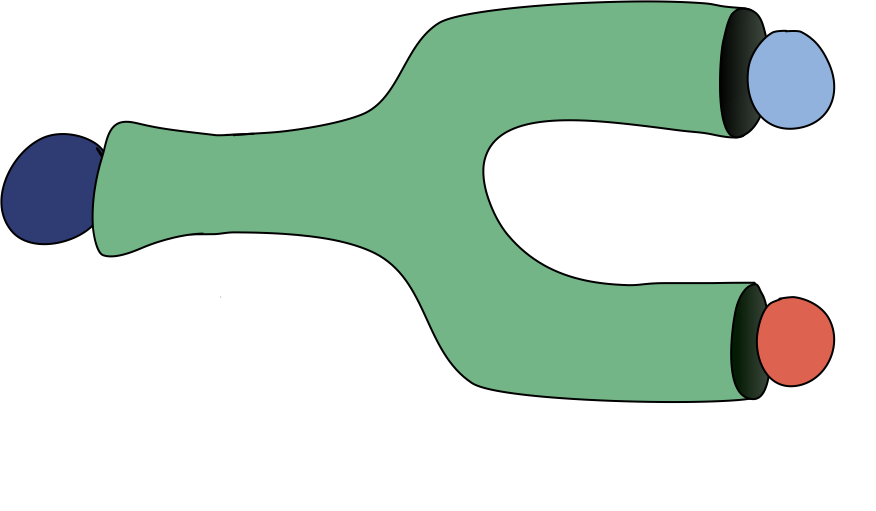} 

    }
        \caption{ \label{fig:tubes} Compatibility of two channels in the Schr\"odinger picture.
        (a) If two channels are implemented together, they require their own inputs. However, (b) a compatible pair can be implemented through their joint channel. In this case, a single input is enough to yield outputs for both channels.}
\end{figure}

The physical idea of compatibility in the Schr\"odinger picture is depicted in Fig. \ref{fig:tubes}.
We will first have two examples and then discuss some basic properties of the compatibility relation.

\begin{example}\label{ex:broad}
(\emph{No universal broadcasting})
 The most paradigmatic example of incompatible channels is the incompatibility of two identity channels. 
 This is nothing else but the impossibility of universal broadcasting; if $\Lambda_1^*$ and $\Lambda_2^*$ are identity channels, then \eqref{eq:marginals-s} is the broadcasting condition of a state $\varrho$ \cite{Barnumetal96}.
 The incompatibility of two identity channels is an obligatory precondition that any pair can be incompatible.
To see this, assume that there is a quantum broadcast channel $\Lambda:\lin(\hi \otimes \hi) \to \lin(\hi)$ that has the identity channels as marginals, i.e., 
\begin{equation}\label{eq:id-marginals}
A = \Lambda(A \otimes \id) \, , \quad B = \Lambda(\id \otimes B)
\end{equation}
for all $A,B\in\lin(\hi)$.
Let $\Lambda_1,\Lambda_2$ be any channels on $\lin(\hi)$.
We concatenate $\Lambda$ with the tensor product channel $\Lambda_1 \otimes \Lambda_2 : \lin(\hi \otimes \hi) \to  \lin(\hi \otimes \hi)$ and then we obtain marginals
\begin{equation}
(\Lambda \circ (\Lambda_1 \otimes \Lambda_2))(A \otimes \id) = \Lambda (\Lambda_1(A) \otimes \id) = \Lambda_1(A) 
\end{equation}
and
\begin{equation}
(\Lambda \circ (\Lambda_1 \otimes \Lambda_2))(\id \otimes B) = \Lambda (\id \otimes \Lambda_2(B)) = \Lambda_2(B) \, . 
\end{equation}
We conclude that if two identity channels are compatible (i.e. there exists $\Lambda$ satisfying \eqref{eq:id-marginals}), then all pairs of channels $\Lambda_1,\Lambda_2$ on $\lin(\hi)$ are compatible.
The impossibility of universal broadcasting is hence equivalent to the statement that there exists a pair of incompatible channels.
\end{example}

\begin{example}\label{ex:noisy}
(\emph{Noise makes channels compatible})
If noise is added enough, then noisy versions of any two channels become compatible.
To see this, let $\Lambda_1:\lin(\hout_1) \to \lin(\hin)$ and $\Lambda_2:\lin(\hout_2) \to \lin(\hin)$ be two channels. 
We fix states $\eta_1\in \state{\hout_1}$, $\eta_2\in \state{\hout_2}$ and define  channels $\Xi_1,\Xi_2$ as
\begin{equation}
\Xi_j(A) = \tr{\eta_j A} \id_{\hin} \, .
\end{equation}
The mixed channels $\half \Lambda_1 + \half \Xi_1$ and $\half \Lambda_2 + \half \Xi_2$ can be seen as noisy versions of $\Lambda_1$ and $\Lambda_2$, respectively.
They are compatible as they have a joint channel
\begin{align}
\Lambda(A \otimes B) = \half \tr{\eta_2 B} \Lambda_1(A) + \half \tr{\eta_1 A}\Lambda_2(B) \, .
\end{align}
This joint channel correspondence to a procedure where we use channels $\Lambda_1$ and $\Lambda_2$ randomly, half of the time each of them.
\end{example}

Example \ref{ex:noisy} was demonstrating the fact that any pair of channels become compatible if they are made noisy enough.
A related fact is that if two channels are compatible, then also channels that are below them in concatenation are compatible.
This is the content of the next proposition.

\begin{proposition}\label{prop:noisy}
Let $\Lambda_1,\Lambda_2,\Phi_1,\Phi_2 \in \chin$ be channels such that $\Phi_1 \pleq \Lambda_1$ and $\Phi_2 \pleq \Lambda_2$.
If $\Lambda_1$ and $\Lambda_2$ are compatible, then also $\Phi_1$ and $\Phi_2$ are compatible.
\end{proposition}

\begin{proof}
By the assumption there are channels $\Theta_1$ and $\Theta_2$ such that $\Phi_1 = \Lambda_1 \circ \Theta_1$ and $\Phi_2 = \Lambda_2 \circ \Theta_2$.
Suppose that $\Lambda_1$ and $\Lambda_2$ are compatible, so they have a joint channel $\Lambda$.
We define a channel $\Phi$ as
\begin{align}
\Phi = \Lambda \circ (\Theta_1 \otimes \Theta_2) \, .
\end{align}
Then 
\begin{align*}
\Phi(A \otimes \id) = \Lambda( \Theta_1(A) \otimes \id ) = \Lambda_1(\Theta_1(A)) = \Phi_1(A)
\end{align*}
and similarly $\Phi(\id \otimes B) = \Phi_2(B)$.
Therefore, $\Phi$ is a joint channel of $\Phi_1$ and $\Phi_2$.
\end{proof}

It follows from Prop. \ref{prop:noisy} that the compatibility relation is the same for all channels that are equivalent in the concatenation sense.
More precisely, we have: 

\begin{corollary}
Let $\Lambda_1,\Lambda_2 \in \chin$ such that $\Lambda_1 \simeq \Lambda_2$.
A channel $\Lambda_3\in \chin$ is compatible with $\Lambda_1$ if and only if it is compatible with $\Lambda_2$.
\end{corollary}

Previous observations show that the compatibility relation is harmoniously connected with the concatenation preorder.
In the following we show that the compatibility relation can, in fact, be characterized in terms of concatenation and conjugation.
Let us first note that for an isometric operator $V$ the map $A \otimes B \mapsto V^*(A\otimes B )V$ is a broadcasting channel, so it follows from the definition that any channel $\Lambda$ and its conjugate channel $\bar{\Lambda}$ are compatible.
The content of the next proposition is that the conjugate channel $\bar{\Lambda}$ is the optimal channel that is compatible with $\Lambda$.
This result can be also taken as the basic characterization of the compatibility relation.
Let us note again that all the conjugate channels of $\Lambda$ are equivalent in the concatenation preorder sense, so the statements (ii) and (iii) in Prop. \ref{prop:compatible} are unambiguous.
The content of Prop. \ref{prop:compatible} is depicted in Fig. \ref{fig:tubes-2}.

\begin{proposition}\label{prop:compatible}
Let $\Lambda_1$ and $\Lambda_2$ be two channels. 
The following are equivalent:
\begin{enumerate}[(i)]
\item $\Lambda_1$ and $\Lambda_2$ are compatible;
\item $\Lambda_1 \pleq \bar{\Lambda}_2$;
\item $\Lambda_2 \pleq \bar{\Lambda}_1$.
\end{enumerate}
\end{proposition}

\begin{proof}
(i)$\Rightarrow$(iii): 
Suppose that $\Lambda_1$ and $\Lambda_2$ are compatible. 
Then there exists a channel $\Lambda$ such that 
$\Lambda(A\otimes \id) = \Lambda_1(A)$ and 
$\Lambda(\id \otimes B) = \Lambda_2(B)$ for all $A\in\lin(\hout_1)$ and $B\in\lin(\hout_2)$. 
Let us fix a Stinespring representation $(V,\hik)$ of $\Lambda$, so that
\begin{equation}
\Lambda(A\otimes B) = V^*(A\otimes B\otimes \id_{\hik})V 
\end{equation}
for all $A\in\mathcal{L}(\hout_1)$ and $B\in\mathcal{L}(\hout_2)$. 
We have
\begin{equation}
\Lambda_1(A) = V^*(A\otimes \id_{\hout_2}\otimes \id_{\hik})V \, ,
\end{equation}
hence $(V,\hik)$ is also a Stinespring representation of $\Lambda_1$.
The conjugate channel $\bar{\Lambda}_1$ related to this representation is written as
\begin{equation}
\bar{\Lambda}_1(B \otimes C) = V^*(\id_{\hout_1} \otimes B \otimes C)V \, .
\end{equation}
We define a channel $\Theta:\lin(\hout_2) \to \lin(\hout_2 \otimes \hik)$ by $\Theta(B ) = B\otimes \id_{\hik}$. 
Then $\bar{\Lambda}_1 \circ \Theta = \Lambda_2$ holds, and hence $\Lambda_2 \pleq \bar{\Lambda}_1$.
\newline
(iii)$\Rightarrow$(ii): Follows from Prop. \ref{prop:c-inverts-order}.
\newline
(ii)$\Rightarrow$(i): By the definition, $\Lambda_2$ and $\bar{\Lambda}_2$ are compatible. Assuming that $\Lambda_1 \pleq \bar{\Lambda}_2$, it follows from Prop. \ref{prop:noisy} that also $\Lambda_1$ and $\Lambda_2$ are compatible.
\end{proof}

\begin{figure}
    \subfigure[]
    {
         \includegraphics[width=5cm]{joint.png}      
         }
    \subfigure[]
    {
        \includegraphics[width=6cm]{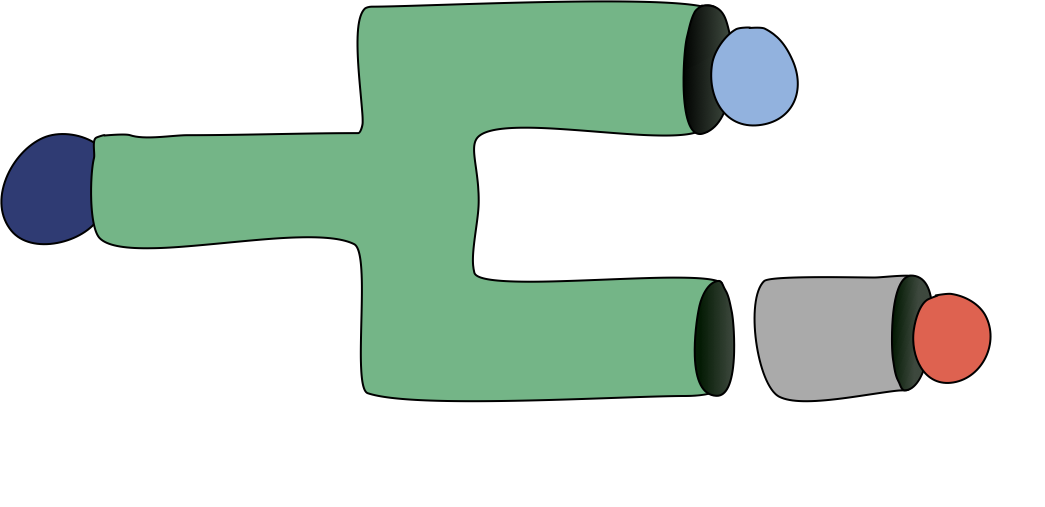} 

    }
        \caption{ \label{fig:tubes-2} The content of Prop. \ref{prop:compatible} in the Schr\"odinger picture. (a) By the definition, a compatible pair of channels has a joint channel. (b) The joint channel can be splitted into an isometric joint channel and local concatenation.}
\end{figure}

\subsection{Joint measurability}

We will next demonstrate that the usual definition of joint measurability of observables is a specific instance of Def. \ref{def:compatible}.
We recall that two observables $\M$ and $\N$ are \emph{jointly measurable} if there exists an observable $\G$ on the Cartesian product set $\Omega_\M \times \Omega_\N$ such that
\begin{equation}\label{eq:joint-obs}
\sum_y \G(x,y) = \M(x) \, , \quad \sum_x \G(x,y) = \N(y) \, ;
\end{equation}
otherwise they are incompatible.
The condition \eqref{eq:joint-obs} just means that the marginals of the probability distribution $\tr{\varrho \G(x,y)}$ are $\tr{\varrho \M(x)}$ and $\tr{\varrho \N(y)}$ for all states $\varrho \in \state{\hin}$.  
The observable $\G$ is called a \emph{joint observable} of $\M$ and $\N$.

\begin{proposition}
Let $\M$ and $\N$ be two observables.
The following are equivalent:
\begin{enumerate}[(i)]
\item $\M$ and $\N$ are jointly measurable;
\item $\Gamma_\M$ and $\Gamma_\N$ are compatible.
\end{enumerate}
\end{proposition} 

\begin{proof}
(i)$\Rightarrow$(ii):
Let $\G$ be a joint observable of $\M$ and $\N$.
We define a channel $\Lambda$ as
\begin{equation}
\Lambda(A) := \sum_{x,y} \ip{\delta_x \otimes \delta_y} {A \delta_x \otimes\delta_y} \G(x,y) \, .
\end{equation}
Then $\Lambda$ is a joint channel of $\Gamma_\M$ and $\Gamma_\N$.

(ii)$\Rightarrow$(i):
Let $\Lambda$ be a joint channel of $\Gamma_\M$ and $\Gamma_\N$.
We define an observable $\G$ as
\begin{equation*}
\G(x,y) := \Lambda(\kb{\delta_x}{\delta_x} \otimes \kb{\delta_y}{\delta_y} ) \, .
\end{equation*}
Then $\G$ is a joint observable of $\M$ and $\N$.
\end{proof}

Let us then note that the content of our earlier Prop. \ref{prop:noisy} can be rephrased as follows: if two channels $\Phi_1$ and $\Phi_2$ are incompatible, then two channels $\Lambda_1,\Lambda_2$ satisfying $\Phi_1 \pleq \Lambda_1$ and $\Phi_2 \pleq \Lambda_2$ are also incompatible.
Since the incompatibility of two observables may be easier to check than the incompatibility of two channels, this implication can be used as a sufficient condition for incompatibility.
Although the following result is a consequence of Prop. \ref{prop:noisy}, we write its short proof explicitly for the sake of clarity.

\begin{proposition}\label{prop:sufficient}
Let $\Lambda_1,\Lambda_2 \in \chin$ and let $\M,\N$ be two observables.
If the observables $\Lambda_1(\M)$ and $\Lambda_2(\N)$ are incompatible, then $\Lambda_1$ and $\Lambda_2$ are incompatible.
\end{proposition}

\begin{proof}
Let us assume that $\Lambda_1$ and $\Lambda_2$ are compatible and let $\Lambda$ be their joint channel.
We define an observable $\G$ on $\Omega_\M \times \Omega_\N$ as
\begin{equation}
\G(x,y) = \Lambda (\M(x) \otimes \N(y) ) \, .
\end{equation}
Then 
\begin{align}
\sum_y \G(x,y)  = \Lambda_1(\M(x)) \, , \quad \sum_x \G(x,y) = \Lambda_2(\N(y)) \, , 
\end{align}
hence, $\G$ is a joint observable of $\Lambda_1(\M)$ and $\Lambda_2(\N)$.
\end{proof}

Using Prop. \ref{prop:sufficient} and some known results for joint measurability of pairs of observables, we can conclude the incompatibility of some pairs of channels.
 The following example demonstrates this kind of reasoning.

\begin{example}(\emph{Incompatible Pauli channels})
Let $\sigma_x,\sigma_y,\sigma_z$ be the usual Pauli operators.
A \emph{Pauli channel} is a channel $\Psi_{\vec{p}}:\lin(\complex^2)\to\lin(\complex^2)$ of the form
\begin{equation}
\Psi_{\vec{p}}(A) = \sum_{j} p_j \sigma_j A \sigma_j + (1-\sum_j p_j) A \, , 
\end{equation}
where $0 \leq p_j \leq 1$ and $\sum_j p_j \leq 1$.
Let $\Psi_{\vec{p}}$ and $\Psi_{\vec{q}}$ be two Pauli channels.
We take two observables $\X$ and $\Y$, defined as
\begin{align}
\X(\pm 1) = \half (\id \pm \sigma_x) \, , \quad \Y(\pm 1) = \half (\id \pm \sigma_y) \, .
\end{align}
We then get
\begin{equation}
\Psi_{\vec{p}}(\X)(\pm 1) =  \half (\id \pm (1-2(p_y+p_z)) \sigma_x) 
\end{equation}
and
\begin{equation}
\Psi_{\vec{q}}(\Y)(\pm 1) =  \half (\id \pm (1-2(q_x+q_z)) \sigma_y) \, .
\end{equation}
As shown in \cite{Busch86} (see \cite{BuHe08} for an alternative proof), the observables $\Psi_{\vec{p}}(\X)$ and $\Psi_{\vec{q}}(\Y)$ are incompatible if and only if
\begin{equation}\label{eq:pauli-ineq}
p_y^2+p_z^2+q_x^2+q_z^2 > \frac{1}{4} \, .
\end{equation}
From Prop. \ref{prop:sufficient} we conclude that two Pauli channels $\Psi_{\vec{p}}$ and $\Psi_{\vec{q}}$ are incompatible whenever the inequality \eqref{eq:pauli-ineq} holds.
We note that the incompatibility of Pauli channels is related to the Pauli cloning of a qubit system \cite{Cerf00}. 
\end{example}

\subsection{Measurement disturbance}

An \emph{instrument} is a map $\I:\Omega\times\lhin \to \lhout$ such that each map $\I(x,\cdot)$ is linear completely positive map and their sum $\sum_x \I(x,\cdot)$ is trace preserving \cite{QTOS76}. 
Any measurement process gives rise to a unique instrument, and an instrument is related to a equivalence class of measurement processes \cite{Ozawa84}.
As in \cite{HeMiRe14}, we say that an observable $\M$ and a channel $\Lambda$ are \emph{compatible} if there is an instrument $\I:\Omega\times\lhin \to \lhout$ such that $\tr{\I(x,\varrho)}=\tr{\varrho \M(x)}$ and $\sum_x \I(x,\varrho)=\Lambda^*(\varrho)$.
This means that $\M$ and $\Lambda$ can describe the same measurement process.
The following observation shows that this usage of the notion compatibility is again consistent with our earlier definition.

\begin{proposition}\label{prop:instrument}
Let $\M$ be an observable and $\Lambda$ a channel.
The following are equivalent:
\begin{enumerate}[(i)]
\item $\M$ and $\Lambda$ are compatible;
\item $\Gamma_\M$ and $\Lambda$ are compatible.
\end{enumerate}
\end{proposition} 

\begin{proof}
(ii)$\Rightarrow$(i):
Let $\Phi$ be a joint channel of $\Gamma_\M$ and $\Lambda$.
We define
\begin{equation}
\I(x,A) := \Phi (\kb{\delta_x}{\delta_x}) \otimes A \, .
\end{equation}
Then $\I$ is an instrument.

(i)$\Rightarrow$(ii):
Let $\I$ be an instrument such that $\tr{\I(x,\varrho)}=\tr{\varrho \M(x)}$ and $\sum_x \I(x,\varrho)=\Lambda^*(\varrho)$. We define
\begin{equation}
\Phi(\kb{\delta_x}{\delta_y} \otimes A) := \delta_{x,y} \I_x(A) \, .
\end{equation}
Then $\Phi$ is a joint channel of $\Gamma_\M$ and $\Lambda$.
\end{proof}

If we aim to measure $\M$ and we want to disturb the system as little as possible, we should choose an instrument $\I$ such that the corresponding channel $\Lambda$ is as high in the concatenation relation as possible. 
By Prop. \ref{prop:instrument} we are searching $\Lambda$ among the channels compatible with $\Gamma_\M$, and by Prop. \ref{prop:compatible} we should thus choose a conjugate channel of $\Gamma_\M$.
To write a conjugate channel for $\Gamma_\M$, we recall that any observable $\M$ has a \emph{Naimark dilation} (see e.g. \cite{CBMOA03}), i.e., a triplet $(\hik, \hat{\M}, K)$ where  $\hik$ is a Hilbert space, $K: \hin \to \hik$ is an isometry and $\hat{\M}$ is a sharp observable on $\hik$ satisfying $K^*\hat{\M}(x)K=\M(x)$ for each $x \in \Omega_{\M}$.
A Stinespring dilation of $\Gamma_\M$ is now obtained by defining an isometry $V$ as
\begin{equation}
V: \hin \to \ltwo{\Omega_\M} \otimes \hik \, , \quad V\psi = \sum_x \delta_x \otimes \hat{\M}(x) K \psi \, .
\end{equation}
The conjugate channel of $\Gamma_\M$ related to this Stinespring dilation, denoted by $\Lambda_\M$, is
\begin{eqnarray}\label{eq:LambdaA-h} 
\Lambda_{\M} (T) =\sum_x  K^\ast \hat{\M}(x) T \hat{\M}(x) K  \, ,
 \end{eqnarray}
 or in the Schr\"odinger picture
\begin{eqnarray}\label{eq:LambdaA-s} 
\Lambda^*_{\M} (\varrho) =\sum_x  \hat{\M}(x) K \varrho K^\ast \hat{\M}(x)  \, .
 \end{eqnarray}
 We say that $\Lambda_\M$ is the \emph{least disturbing channel} for $\M$.

Combining these observations with Prop. \ref{prop:c-inverts-order} and Prop. \ref{prop:compatible}, we have recovered the \emph{qualitative noise-disturbance relation}, first presented in \cite{HeMi13}. 

\begin{corollary}\label{cor:noise-dist}
Let $\M$ and $\N$ be two observables.
The following are equivalent:
\begin{enumerate}[(i)]
\item $\N$ is a post-processing of $\M$;
\item  $\Lambda_\M \pleq \Lambda_\N$;
\item If a channel $\Lambda$ is compatible $\M$, then it is also compatible with $\N$.
\end{enumerate}
\end{corollary}

The message of this result is that if we measure a noisier observable instead of a sharper one, then we can choose a measurement process that disturbs the input state less.
This qualitative statement is meaningful even without a specific quantification of disturbance since it connects to the whole sets of compatible channels of the compared observables.

\subsection{Incompatibility of several channels}\label{sec:more}

A quantum broadcast channel can have a total output space which is a tensor product of not only two output spaces but many of them.
This generalization leads to some new aspects.
The concepts of compatibility and incompatibility have the following direct generalizations for any finite number of channels. 

\begin{definition}\label{def:compatible}
Let $\Lambda_j:\lin(\hout_j) \to \lin(\hin)$, $j=1,\ldots,n$, be channels.
If there exists a channel 
\begin{equation}
\Lambda:\lin(\hout_1 \otimes \hout_2 \otimes \cdots \otimes \hout_n)\to\lin(\hin)
\end{equation}
 such that 
 \begin{equation}\label{eq:n}
 \begin{aligned}
\Lambda_1(A_1) &= \Lambda(A_1 \otimes \id \otimes \cdots \otimes \id)  \\
\Lambda_2(A_2) &= \Lambda(\id \otimes A_2 \otimes \cdots \otimes \id) \\
&\vdots \\
\Lambda_n(A_n) &= \Lambda(\id \otimes \id \otimes \cdots \otimes A_n)
\end{aligned}
\end{equation}
for all $A_j \in \lin(\hout_j)$, then $\Lambda_1,\ldots,\Lambda_n$ are \emph{compatible} and $\Lambda$ is their \emph{joint channel}.
Otherwise $\Lambda_1,\ldots,\Lambda_n$ are \emph{incompatible}.
\end{definition}

The following is a direct generalization of Prop. \ref{prop:noisy}.
The proof is similar and we thus omit it.

\begin{proposition}\label{prop:noisy-2}
Let $\Lambda_j,\Phi_j \in \chin$, $j=1,\ldots,n$, be channels such that $\Phi_j \pleq \Lambda_j$ for every $j=1,\ldots,n$.
If $\Lambda_1,\ldots,\Lambda_n$ are compatible, then also $\Phi_1,\ldots,\Phi_n$ are compatible.
\end{proposition}

As was shown in Prop. \ref{prop:compatible}, the compatibility of two channels has a neat characterization in terms of the concatenation preorder and conjugation. 
As expected, concatenation and conjugation are still closely related to the compatibility of more than two channels.
However, the characterization is now more involved.

\begin{proposition}\label{prop:three}
For three channels $\Lambda_1$, $\Lambda_2$ and $\Lambda_3$, the following statements are equivalent. 
\begin{itemize}
\item[(i)] 
$\Lambda_1$, $\Lambda_2$ and $\Lambda_3$ are compatible.
\item[(ii)]
There exist compatible channels $\mathcal{E}_2$ and $\mathcal{E}_3$ 
such that
$\Lambda_2 = \bar{\Lambda}_1 \circ \mathcal{E}_2$
and $\Lambda_3 = \bar{\Lambda}_1 \circ \mathcal{E}_3$.  
\end{itemize} 
\end{proposition}

\begin{proof}
(ii) $\Rightarrow$ (i):
As $\mathcal{E}_2 $ and $\mathcal{E}_3$ are compatible,
there exists a channel $\Phi$ such that 
$\Phi(B \otimes \id) = \mathcal{E}_2(B)$ and 
$\Phi(\id \otimes C) = \mathcal{E}_3(C)$. 
Let us consider a Stinespring dilation of $\Lambda_1$ described by $(V, \hik)$. 
Then we can define a channel $\Lambda$ as 
\begin{equation}
\Lambda(A\otimes B \otimes C)
:= V^* (A \otimes \Phi(B \otimes C))V \, .
\end{equation}
The marginals of $\Lambda$ coincide with $\Lambda_1$, $\Lambda_2$, and $\Lambda_3$. 
\\
(i) $\Rightarrow$ (ii):
There is a channel $\Lambda$ such that the relevant conditions \eqref{eq:n} hold. 
Let us consider a channel $\Lambda_{23}$ defined by 
$\Lambda_{23}(B \otimes C) = \Lambda(\id \otimes B \otimes C)$.
It follows that $\Lambda_{23}$ and $\Lambda_1$ are compatible, 
hence $\Lambda_{23} \pleq \bar{\Lambda}_1$.
Therefore, there exists a channel $\Phi$ such that,  by using a Stinespring representation of $\Lambda_1$, we can write
\begin{equation}
\Lambda(\id\otimes B \otimes C) = V^*(\id \otimes \Phi(B \otimes C) )V \, .
\end{equation}
Defining $\mathcal{E}_2(B) := \Phi(B \otimes \id)$ 
and $\mathcal{E}_3(C):= \Phi(\id \otimes C)$, 
we obtain (ii). 
\end{proof}

The physical content of Prop. \ref{prop:three} is illustrated in Fig. \ref{fig:three}.

\begin{figure}
    \subfigure[]
    {
         \includegraphics[height=3cm]{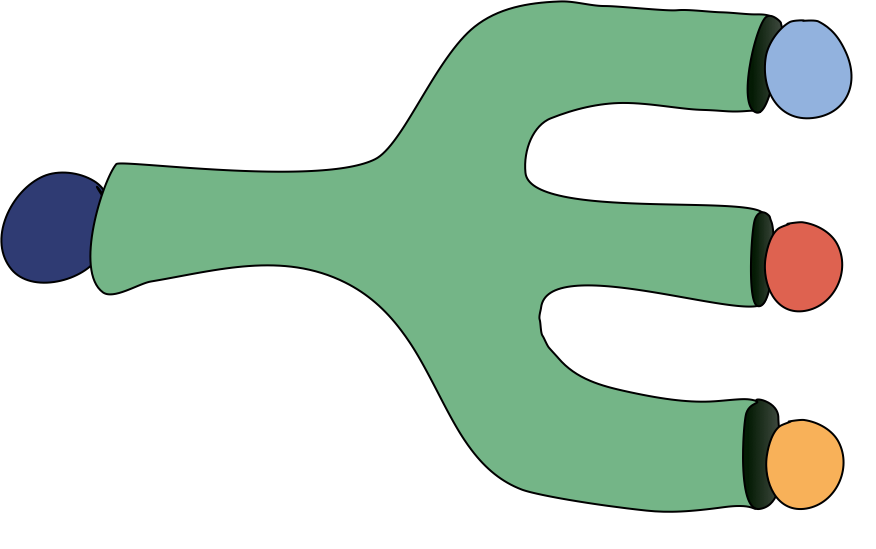}      
         }
    \subfigure[]
    {
        \includegraphics[height=3cm]{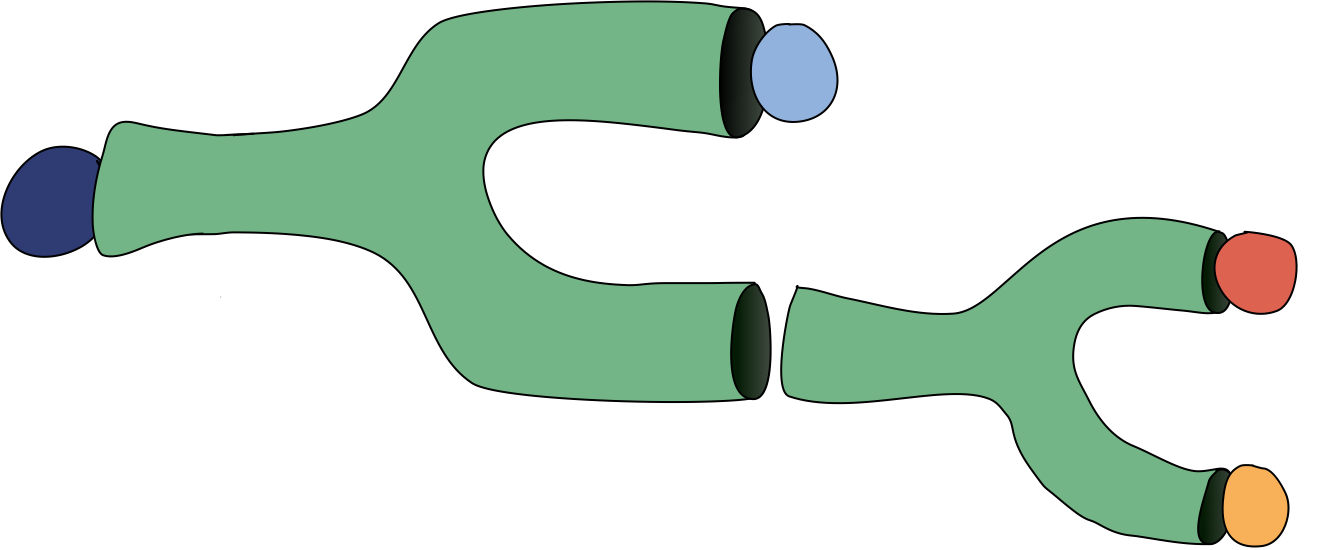} 

    }
        \caption{ \label{fig:three} Compatibility of three channels in the Schr\"odinger picture.
        (a) By the definition, there is a broadcast channel with three outputs. (b) By Prop.\ref{prop:three}, there exist two broadcast channels with two outputs that are combined together.}
\end{figure}

\section{Special types of channels}\label{sec:special}

\subsection{Completely depolarizing channels}

We recall that a channel $\Phi^*:\state{\hin} \to \state{\hout}$ is called \emph{completely depolarizing} if there is a state $\eta\in\state{\hout}$ such that 
\begin{equation}\label{eq:cd-s}
\Phi^*(\varrho)= \tr{\varrho} \eta
\end{equation}
for all input states $\varrho\in\state{\hin}$.
In the Heisenberg picture this formula reads
\begin{equation}\label{eq:cd}
\Phi(A)= \tr{\eta A} \id \, .
\end{equation}
We can easily see that any two completely depolarizing channels 
are concatenation equivalent. Namely, for any two states $\eta_1$ and $\eta_2$, 
there exists a channel satisfying $\mathcal{E}^*(\eta_1) 
= \eta_2$. 
This gives $\tr{\eta_2 A} \id = \tr{\eta_1 \mathcal{E}(A)} \id$. 
Moreover, any channel which is concatenation equivalent with a completely depolarizing channel is also a completely depolarizing channel; if $\Gamma$ satisfies 
$\Gamma \pleq \Phi$ for a completely depolarizing channel, 
there exists a channel $\mathcal{E}$ satisfying 
\begin{equation}
\Gamma (A)= \Phi \circ \mathcal{E}(A)
= \tr{ \mathcal{E}^* (\eta) A} \id \, .
\end{equation}
This calculation also shows that the equivalence class consisting of all completely depolarizing channels is the smallest element in the preordered set $\chan$.   
 
By the definition, the output state of a completely depolarizing channel $\Phi^*$ does not depend on the input state at all.
Therefore, it is evident that $\Phi$ is compatible with any other channel $\Lambda \in \chin$.
A joint channel for $\Phi$ and $\Lambda$ is 
\begin{equation}
A \otimes B \mapsto \Phi(A) \Lambda(B) \, .
\end{equation}
This property of completely depolarizing channels is their characteristic feature; any channel compatible with all channels is completely depolarizing.

\begin{proposition}
The following statements for a channel $\Phi$ are equivalent. 
\begin{itemize}
\item[(i)] $\Phi$ is completely depolarizing.
\item[(ii)] $\Phi \pleq \Gamma$ for any channel $\Gamma$.  
\item[(iii)] $\Phi$ is compatible with any other channel $\Lambda$.
\item[(iv)]$\Phi$ is compatible with the identity channel $id$.
\end{itemize}
\end{proposition}

\begin{proof}
(i) $\Rightarrow$ (ii). 
Trivial.
(ii)$\Rightarrow$ (iii):
Follows from Prop. \ref{prop:compatible}. 
(iii)$\Rightarrow$ (iv):
Trivial.
(iv) $\Rightarrow$ (i): 
Let $\Phi$ be a channel which is compatible with the identity channel $id$. 
By Prop. \ref{prop:sufficient}, for any pair of observables $\M$ and $\N$, then transformed observables $\Phi(\M)$ and $id(\N)=\N$
must be compatible. 
Thus for an operator $A\in\lin(\hout)$, the operator $\Phi(A)$ commutes with all projections $B\in\lin(\hout)$.
This implies that $\Phi(A)$ must be a scalar multiple of the identity operator. 
Since this is true for all $A\in\lin(\hout)$, there exists a state $\eta\in\state{\hout}$ such that 
$\Phi(A) = \tr{\eta A} \id$.   
\end{proof}

\subsection{Entanglement breaking channels}

We recall that a channel $\Phi^*$ is called \emph{entanglement breaking} if the bipartite state $(\Phi^* \otimes id)(\omega)$ is separable for any choice of the input state $\omega$.
This is equivalent to the condition that $\Phi^*$ is of the measure-prepare form \cite{HoShRu03}, i.e., there exists an observable $\F$ on $\hik$ and a set of states $\{\eta_x\}\subset\state{\hin}$ such that 
\begin{equation}\label{eq:ebc-s}
\Phi^*(\varrho)= \sum_x \tr{\varrho \F(x)} \eta_x \, .
\end{equation}
In the Heisenberg picture this reads
\begin{equation}\label{eq:ebc}
\Phi(A)= \sum_x \tr{\eta_x A} \F(x) \, .
\end{equation}

If $\Phi$ is entanglement breaking, then for any channel $\Lambda \in \chin$, the concatenated channels $\Phi \circ \Lambda$ are $\Lambda \circ \Phi$ are still entanglement breaking.
Namely, we have
\begin{equation}
(\Phi \circ \Lambda) (A)= \sum_x \tr{\Lambda^*(\eta_x) A} \F(x) \, ,
\end{equation}
and
\begin{equation}
(\Lambda \circ \Phi) (A)= \sum_x \tr{\eta_x A} \Lambda(\F(x)) \, ,
\end{equation}
which are both of the measure-prepare form.

\begin{proposition}\label{prop:below-M}
Let $\M$ be an observable and $\Lambda: \lin(\hout) \to \lin(\hin)$ a channel.
The following are equivalent:
\begin{enumerate}[(i)]
\item $\Lambda \pleq \Gamma_\M$
\item 
There exists a family of states $\{\eta_x\} \subset \state{\hout}$ such that 
\begin{equation}\label{eq:measure-prepare-M}
\Lambda(A)= \sum_x \tr{\eta_x A} \M(x) \, .
\end{equation}
\end{enumerate}
In particular, a channel $\Lambda$ is entanglement breaking if and only if $\Lambda \pleq \Gamma_\M$ for some observable $\M$.
\end{proposition}

\begin{proof}
(i)$\Rightarrow$(ii): There exists a channel $\mathcal{E}$ such that 
$\Lambda = \Gamma_{\M}\circ \mathcal{E}$. 
We define $\eta_x$ as $\rho_x := \mathcal{E}^*(\kb{\delta_x}{\delta_x})$, 
and then
\begin{eqnarray*}
\Lambda(A) = (\Gamma_{\M}\circ \mathcal{E})(A)= \sum_x \tr{\eta_x A} \M(x). 
\end{eqnarray*}
\newline
(ii)$\Rightarrow$(i): We define a channel $\mathcal{E}$ by 
\begin{equation}
\mathcal{E}(A):= \sum_x \tr{\eta_x A } \kb{\delta_x}{\delta_x} \, .
\end{equation} 
Then it gives $\Lambda = \Gamma_{\M}\circ \mathcal{E}$. 
\end{proof}

\subsection{Self-compatible channels}

We recall that a channel $\Lambda$ is called \emph{self-compatible} if $\Lambda$ is compatible with itself \cite{Haapasalo15}.
By the definition, a channel $\Lambda$ is self-compatible if there exists a broadcast channel that simulates the output of $\Lambda$ twice for a single input.
In this sense, the action of a self-compatible channel can be duplicated. 
By Prop. \ref{prop:compatible} a channel $\Lambda$ is self-compatible if and only if $\Lambda \pleq \bar{\Lambda}$.
Therefore, \emph{the self-compatible channels are exactly the antidegradable channels}. 
This equivalence leads to some useful observations.
For instance, the fact that the set of antidegradable channels (with fixed input and output spaces) is convex \cite{CuRuSm08} follows directly from our our framework; it is easy to see from the definition of compatibility that a convex combination of two self-compatible channels is again self-compatible.
We also recall that the antidegradable channels have been characterized in a game-theoretic framework \cite{BuDaSt14}, and this further clarifies the meaning of self-compatibility.

In Sec. \ref{sec:more} we have defined the compatibility of $n$ channels, so we can also ask if $n$ copies of a given channel are compatible.
This leads to the following notion. 

\begin{definition}
A channel $\Lambda$ is \emph{$n$-self-compatible} if $n$ copies of $\Lambda$ are compatible.
\end{definition}

It is easy to see that \emph{every observable is $n$-self-compatible for any $n$}.
Physically the reason is simply that we can copy the obtained measurement outcomes.
To see this in our mathematical formalism, we fix a finite set $\Omega\subset\integer$ and define the copying channel $\Cc^*_n : \state{\ltwo{\Omega}} \to \state{\ltwo{\Omega})^{\otimes n}}$ of the orthonormal basis $\{\delta_x\}_{x \in \Omega}$ as
\begin{equation}
\Cc^*_n (\varrho) = \sum_x \ip{\delta_x}{\varrho \delta_x} \kb{\delta_x^{\otimes n}}{\delta_x^{\otimes n}} \, .
\end{equation}
For each observable $\M$ with the outcome set $\Omega$, we have
\begin{align*}
(\Cc^*_n \circ \Gamma^*_\M) (\varrho) = \sum_x \tr{\varrho \M(x)} \kb{\delta_x^{\otimes n}}{\delta_x^{\otimes n}} \, .
\end{align*}
It is straightforward to verify that $\Cc^*_n \circ \Gamma^*_\M$ is a joint channel for $n$ copies of $\Gamma^*_\M$.

\begin{proposition}\label{prop:n-self}
A channel $\Lambda$ is $n$-self-compatible for all $n=2,3,\ldots$ if and only if $\Lambda$ is entanglement breaking.
\end{proposition}

\begin{proof}
'If': Let $\Lambda$ be an entanglement breaking channel.
By Prop. \ref{prop:below-M} there exists an observable $\M$ such that $\Lambda \pleq \Gamma_\M$.
As we have seen, $\Gamma_\M$ is $n$-self-compatible for every $n$.
It follows from Prop. \ref{prop:noisy-2} that also $\Lambda$ is $n$-self-compatible for every $n$.
\newline
'Only if:' Let $\Lambda$ be a channel that is $n$-self-compatible for all $n=2,3,\ldots$.
Fix $n$, and let $\Theta_n$ be a joint channel for the $n$ copies of $\Lambda$.
We define a modified channel $\widetilde{\Theta}_n$ as
\begin{equation}
\widetilde{\Theta}_n (A) := \frac{1}{n!} \sum_{\pi \in S_n}\Theta_n(U^{(n)*}_\pi A U^{(n)}_\pi ) \, , 
\end{equation}
where $S_n$ is the symmetric group of all permutations of $n$ objects and $U^{(n)}_\pi$ is the unitary operator on $\hi^{\otimes n}$ that permutes the $n$ copies of $\hout$ according to the permutation $\pi\in S_n$.
The channel $\widetilde{\Theta}_n$ is still a joint channel for the $n$ copies of $\Lambda$, and it satisfies the additional symmetry condition
\begin{equation}
\widetilde{\Theta}_n (A) = \widetilde{\Theta}_n(U^{(n)*}_\pi A U^{(n)}_\pi ) 
\end{equation}
for all $A$ and $\pi\in S_n$.
This symmetry property means that $\widetilde{\Theta}_n$ is a symmetric broadcast channel.
It was proved in \cite{Chiribella11} that there exists an entanglement breaking channel  $\Phi_n$, depending on $n$, such that
\begin{align}
\no{ \Lambda^* -  \Phi_n^* }_\diamond \leq \frac{2 d^2}{n} \, .
\end{align}
As this is true for all $n$, we conclude that for $\Lambda$ is arbitrarily closed to an entanglement breaking channel. 
Since the set of entanglement breaking channels is closed, $\Lambda$ must be an entanglement breaking channel itself.
\end{proof}

\subsection{Incompatibility breaking channels}

As defined in \cite{HeKiReSc15}, a channel $\Phi$ is \emph{$n$-incompatibility breaking} if observables $\Phi(\M_1),\ldots,\Phi(\M_n)$ are jointly measurable for any choice of $n$ observables $\M_1,\ldots,\M_n$.
As we have seen, the joint measurability of observables is equivalent to the compatibility of the respective channels. 
Therefore, a channel $\Phi$ is $n$-incompatibility breaking if and only if the channels $\Phi \circ \Gamma_{\M_1},\ldots,\Phi \circ \Gamma_{\M_n}$ are compatible for all observables $\M_1,\ldots,\M_n$.

The set $\{ \Gamma_\M: \M \in \obsin \}$ is a subset of $\chin$.
We have thus the following direct generalization of the notion of $n$-incompatibility breaking channels to an arbitrary subset of $\chin$.

\begin{definition}
Let $\Phi: \lin(\hin) \to \lin(\hik)$ be a channel and $\chansub \subseteq \chin$.
\begin{enumerate}[(a)]
\item For $n=2,3,\ldots$, $\Phi$ is \emph{$n$-incompatibility breaking on $\chansub$} if channels $\Phi \circ \Lambda_1,\ldots,\Phi \circ \Lambda_n$ are compatible for all channels $\Lambda_1,\ldots,\Lambda_n \in \chansub$.
\item $\Phi$ is \emph{incompatibility breaking on $\chansub$} if it is $n$-incompatibility breaking on $\chansub$ for all $n$.
\end{enumerate}
\end{definition}

It has been demonstrated in \cite{HeKiReSc15,Pusey15}, that the set of incompatibility breaking channels on $\obsin$ includes all entanglement breaking channels but it also includes other kind of channels.
In contrast, the next result shows that if we consider the incompatibility breaking channels on the total set $\chin$, these are just the entanglement breaking channels.

\begin{proposition}
A channel $\Phi: \lin(\hin) \to \lin(\hik)$ is incompatibility breaking on $\chin$ if and only if it is entanglement breaking.
\end{proposition}

\begin{proof}
'If': Let $\Phi$ be an entanglement breaking channel.
Hence, it can be written as in \eqref{eq:ebc} for some observable $\F$ on $\hik$ and a set of states $\{\eta_x\}\subset\state{\hin}$.
For a collection of $n$ channels $\Lambda_1,\ldots,\Lambda_n \in \chin$, we define a set of states $\tilde{\eta}_x$ as
\begin{align}
\tilde{\eta}_x = \Lambda_1^*(\eta_x) \otimes \cdots \otimes  \Lambda_n^*(\eta_x) \, .
\end{align}
We then define a channel 
$\Lambda:\lin(\hout_1 \otimes \hout_2 \otimes \cdots \otimes \hout_n)\to\lin(\hin)$
as
\begin{equation}
\Lambda(A) = \sum_x \tr{\tilde{\eta}_xA} \F(x) \, .
\end{equation}
It is straightforward to verify that $\Lambda$ is a joint channel for the channels $\Phi \circ \Lambda_1,\ldots,\Phi\circ \Lambda_n$.
Hence, $\Phi$ is $n$-incompatibility breaking on $\chin$ for $n$.
\newline
'Only if:' Fix $n=2,3,..$ and let $id:\lin(\hin)\to\lin(\hin)$ be the identity channel.
As $\Phi \circ id = \Phi$ and $\Phi$ is assumed to be $n$-incompatibility breaking, we conclude that $\Phi$ is $n$-self-compatible.
This is true for all $n$, so it follows from Prop. \ref{prop:n-self} that $\Phi$ is entanglement breaking.
\end{proof}

\paragraph{\emph{Acknowledgments}.}

The authors are grateful to Giulio Chiribella, Erkka Haapasalo, Yui Kuramochi, Jussi Schultz and Mario Ziman for their comments on an earlier version of this paper. 
T.H. acknowledges support from the Academy of Finland (Project No. 287750).


\begin{thebibliography}{10}
\bibitem{Plavala16}
M.~Pl\'avala.
\newblock All measurements in a probabilistic theory are compatible if and only
  if the state space is a simplex.
\newblock {\em Phys. Rev. A}, 94:042108, 2016.
\end{thebibliography}
\end{document}